\documentclass[runningheads]{llncs}
\usepackage{graphicx}
\usepackage[utf8]{inputenc}
\usepackage{amsmath,amssymb}
\usepackage{pstricks}
\usepackage{float}

\title{Boundary-type Sets of Strong Product of Directed Graphs}
\titlerunning{Boundary-type Sets\ldots}  
\author{Prasanth G. Narasimha-Shenoi\inst{1} \and
Bijo S Anand\inst{2} \and
Mary Shalet T J\inst{1}}
\authorrunning{P.G. Narasimha-Shenoi et al.} 
\institute{Department of Mathematics,
	Government College Chittur,\\
	Palakkad, India - 678104\\
	\email{prasanthgns@gmail.com, mary\_shallet@yahoo.co.in}
\and
Sree Narayana College, Punalur, Kollam, Kerala, India
,\\
\email{bijos\_anand@yahoo.com}
}

\date{}

\begin{document}

\maketitle              

\begin{abstract}
Let $D=(V,E)$ be a strongly  connected digraph and let $u ,v\in  V(D)$.  The maximum distance $md (u,v)$ is defined as\\ $md(u,v)$=max\{$\overrightarrow{d}(u,v), \overrightarrow{d}(v,u)$\} where $\overrightarrow{d}(u,v)$ denote the length of a shortest directed $u-v$ path in $D$.  This is a metric.  The boundary, contour, eccentric and peripheral sets of a strong digraph $D$ with respect to this metric have been defined, and the above said metrically defined sets of a large strong digraph $D$ have been investigated in terms of the factors in its prime factor decomposition with respect to Cartesian product.  In this paper we investigate about the above boundary-type sets of a strong digraph $D$ in terms of the factors in its prime factor decomposition with respect to strong product. 
\end{abstract}

\section{Introduction}
Directed graphs or in short digraphs have immense applications in almost all areas of science and even in  sociology.  Nowadays, one-way networks are introduced almost everywhere with an objective of increasing the efficiency of the network. A directed network is a network in which each edge has a direction, pointing from one vertex to another.  They can be represented as directed graphs.  Dealing with directed networks is more complicated than dealing with two-way networks.\\
  Road traffic networks are the most frequently met examples of one-way network in day to day life.  Almost all main road networks are now kept as one-way.  The reason for this is the decreased accident rate and the ease of driving in one-way roads.  But when one-way traffic is introduced in a two-way network, there arises the difficulty of increased distance between places in one of the directions.  So here the problem of designing the network so as to minimize the distance between places as well as decreasing the cost of construction comes to play.\\
     The one-way problem was first studied by Robbins \cite{robbins1939theorem}. They have applications in a variety of fields like computer science and biolgy.  In the case of internet, the structure of the network will affect how efficiently it accomplishes its function of transmitting data.  If we know the network structure we can address many questions of practical relevance.  For example, if we have the opportunity to add new capacity to the network, then we can determine where should it be added.  Molecular biologists use networks to represent the patterns of chemical reactions among chemicals in the cell, while neuroscientists use them to represent patterns of connections between brain cells.  The other networks that we come across in our daily life are telephone networks, the power grid and email networks \cite{net}.\\
   The boundary-type sets of a graph, the \textit{boundary, contour, eccentricity} and \textit{periphery sets} of a graph were studied in \cite{caceres2005rebuilding} and \cite{chartrand2003boundary}.  They  constitute the borders of a graph.  All other vertices of the graph lie between them.  The efficiency of a network is defined in terms of how easy it is for pairs of nodes to communicate with each other.  So it is apt to say that the boundary-type vertices determine the efficiency of a network.\\
 It is very difficult to  identify the various boundary-type sets in large networks.  So naturally we try to decompose the network into smaller networks and identify the boundary-type sets.  The four standard graph products are the Cartesian, the direct, the strong and the lexicographic product.  They can be extended to digraphs as well.  Thus our strategy is to apply the technique of \textit{Divide and Rule} in order to identify the boundary-type sets of large networks.\\
 Feigenbaum showed that the Cartesian product of digraphs satisfies the unique prime factorization property and provided a polynomial time algorithm for its computation \cite{feigenbaum1986directed}.  This was improved to a linear time approach by Crespelle et al. \cite{crespelle}.  Marc Hellmuth and Tilen Marc developed a polynomial time algorithm for determining the prime factor decomposition of digraphs with respect to the strong product \cite{hellmuth}.\\ 
 The usual directed distance in digraphs is not a metric.  For any two vertices $u$ and $v$ in a strong digraph $D$, the directed distance $\overrightarrow{d}(u,v)$ is usually not the same as the directed distance $\overrightarrow{d}(v,u)$.   As we are concerned with the problem of designing the network so as to minimize the distance between places at a minimum cost, we consider the distance \textit{maximum distance} or in short \textit{m-distance} which is a metric that was introduced by Chartrand and Tian in \cite{chartrand1997distance}.  It gives the maximum of the directed distances in either direction and is denoted by $md(u,v)$.  So minimizing $md(u,v)$ actually results in minimizing the distance between the nodes in both directions.  We can see that the metric \textbf{md} is a generalisation of the usual distance metric in undirected graphs.  The \textit{m-eccentricity} $me(v)$ of a vertex $v$ of $D$ is defined as $me(v)=\max_{u \in V(D)}\{md(v,u)\}$, the \textit{m-radius} of digraph $D$ is $m rad (D)=\min_{v \in V(D)}\{me(v)\}$, and the \textit{m-diameter} is $m diam (D)=\max_{v \in V(D)}\{me(v)\}$.\\
 The results concerning the boundary-type sets of a digraph in terms of its factors in the prime factor decomposition with respect to Cartesian product was presented in \cite{cartesian}.  In this paper, a similar study is conducted for strong product.
\section{Preliminaries}
A \textit{directed graph} or a \textit{digraph} $D$ consists of a non-empty finite set $V(D)$ of elements called vertices and a finite set $E(D)$ of ordered pairs of distinct vertices called arcs or edges \cite{bang2008digraphs}.   We call $V(D)$ the vertex set and $E(D)$ the edge set of $D$.  We write $D=(V, E)$ to denote the digraph $D$ with vertex set $V$ and edge set $E$.  For an edge $(u,v)$, the first vertex $u$ of the ordered pair is the tail of the edge and the second vertex $v$ is the head; together they are the endpoints.  This definition of a digraph does not allow loops (edges whose head and tail coincide) or parallel edges (pairs of edges with same tail and same head).\\
A \textit{directed path} is a directed graph $P\neq \emptyset$  with distinct vertices $u_0,\ldots , u_k$ and edges $e_0, \ldots , e_{k-1}$ such that $e_i$ is an edge directed from $u_i$ to $u_{i+1}$, for all $i<k$.  In this paper, a path will always mean a `directed path'.  A digraph is \textit{strongly connected} or\textit{ strong} if, for each ordered pair $(u, v)$ of vertices, there is a path from $u$ to $v$.  A digraph is \textit{weakly connected} if its underlying graph is connected.

The \textit{length} of a path is the number of edges in the path.  Let $u$ and $v$ be vertices of a strongly connected digraph $D$.  A shortest directed $u-v$ path is also called a directed $u-v$ \textit{geodesic}.  The number of edges in a  directed $u-v$ geodesic is called the directed distance $ \overrightarrow{d}(u,v)$.   But this distance is not a metric because $ \overrightarrow{d}(u,v)\neq \overrightarrow{d}(v,u)$.  So in \cite{chartrand1997distance}, Chartrand and Tian introduced two other distances in a strong digraph, namely the maximum distance $md(u,v)=max\{\overrightarrow{d}(u,v),\overrightarrow{d}(v,u)\}$ and the sum distance $sd(u,v)=\overrightarrow{d}(u,v)+\overrightarrow{d}(v,u)$, both of which are metrics.  In this paper, we deal with the first metric, the maximum distance, $md$.\\ 
Hereafter, we denote  $md(u,v)$ by $d(u,v)$, $me(v)$ by $ecc(v)$, $mrad(D)$ by $rad(D)$ and $mdiam(D)$ by $diam(D)$.  Also we consider only strong digraphs so that the distance between vertices and eccentricity of a vertex are always defined.

The concept of neighbourhood in a digraph $D$ is as follows \cite{bang2008digraphs}.\\
$N_D^+(v)=\{u \in V-v:vu \in E\}$, $N_D^-(v)=\{w \in V-v:wv \in E\}$.  The sets $N_D^+(v)$, $N_D^-(v)$ and $N_D(v)=N_D^+(v) \bigcup N_D^-(v)$ are called the out-neighbourhood, in-neighbourhood and neighbourhood of $v$.
   The closed neighbourhood (neighbours including $v$) of $v$ is denoted by  $N[v]$.  
\subsection{Definitions of boundary-type sets}
We define the boundary-type sets of a digraph $D$ with respect to the metric \textbf{maximum distance}.   
Most of the following definitions are analogous to the definitions in  \cite{chartrand2003boundary}.  Let $D$ be a strong digraph and $u, v \in V(D)$.  The vertex $v$ is said to be a \textit{boundary vertex} of $u$ if no neighbour of $v$ is further away from $u$ than $v$.  A vertex $v$ is called a \textit{boundary vertex} of $D$ if it is the boundary vertex of some vertex $u \in V(D)$. 
\begin{definition}
	The \textit{boundary} $\partial(D)$ of $D$ is the set of all of its boundary vertices; that is  $\partial(D) = \{v \in V | \exists u \in V, \forall w \in N(v) : d(u,w)\leq d(u, v)\}$. 
\end{definition}
The \textit{eccentricity} of a vertex $u \in V(D)$ is defined as $ecc_D(u) = max\{d(u, v)| v \in V(D)\}$.  If the digraph $D$ is clear from the context, we denote it as $ecc(u)$.  Given $u, v \in V(D)$, the vertex $v$ is called an \textit{eccentric vertex} of $u$ if no vertex in $V(D)$ is further away from $u$ than $v$.  This means that $d(u, v) = ecc(u)$.  A vertex $v$ is called an \textit{eccentric vertex} of digraph $D$ if it is the eccentric vertex of some vertex $u \in V(D)$. 
\begin{definition}
	The \textit{eccentricity} $Ecc(D)$ of a digraph $D$ is the set of all of its eccentric vertices;  
	$Ecc(D) = \{v \in V(D) | \exists u \in V(D), ecc(u) = d(u, v)\}$.
\end{definition}
A vertex $v \in V(D)$ is called a \textit{peripheral vertex} of digraph $D$ if no vertex in $V(D)$ has an eccentricity greater than $ecc(v)$,  that is, if the
eccentricity of $v$ is exactly equal to the diameter $diam (D)$ of $D$. 
\begin{definition}
	The \textit{periphery} $Per(D)$ of a digraph $D$ is the set of all of its peripheral vertices;  
	$Per(D) = \{v \in V(D) | ecc(u)\leq ecc(v), \forall u \in V(D)\}$. \\
	That is, $Per(D)= \{v \in V(D) | ecc(v) =
	diam(D)\}$.  
\end{definition}
A vertex $v \in V(D)$ is called a \textit{contour vertex} of digraph $D$ if no neighbour vertex of $v$ has an eccentricity greater than $ecc(v)$. The following
definition is from \cite{caceres2005rebuilding}.
\begin{definition}
	The \textit{contour} $Ct(D)$ of a digraph $D$ is the set of all of its contour vertices; 
	$Ct(D) = \{v \in V(D) | ecc(u)\leq ecc(v), \forall u \in N(v)\}$.
\end{definition} 
Then it is obvious from the definitions that as in the case of undirected graphs as in \cite{caceres2006geodetic} we have, 
\begin{enumerate}
	\item
	$Per(D) \subseteq Ct(D)\cap Ecc(D)$.
	\item
	$ Ecc(D) \cup Ct(D) \subseteq \partial(D)$.
\end{enumerate}
In the above definitions, $N(v)$ can also be replaced by $N[v]$.
\section{Strong Product of Directed Graphs}
The strong product 
$D_1 \boxtimes D_2$ of digraphs $D_1$ and $D_2$ is the digraph having vertex set
$V(D_1) \times V(D_2)$ and with arc set defined as follows.  A vertex $(u_i,v_r)$ is adjacent to $(u_j,v_s)$ in $D_1 \boxtimes D_2$ if either 
\begin{enumerate}
	\item
	$(u_i,u_j) \in E(D_1)$, $v_r=v_s$, or
	\item
	$u_i=u_j$, $(v_r,v_s) \in E(D_2)$, or
	\item
	$(u_i,u_j) \in E(D_1)$, $(v_r,v_s) \in E(D_2).$
\end{enumerate}
The strong product of digraphs is commutative \cite{hammack18}.   The distance between two vertices $(g,h)$ and $(g',h')$ in the strong product $G \boxtimes H$ of two graphs $G$ and $H$ is given in \cite{handbook}.  The relationship between 
 distance in the strong product and distances in its factor graphs is as follows:
$d_{G \boxtimes H}((g,h),g',h'))=\max\{d_G(g,g'),d_H(h,h')\}$.  So for the  strong product $D_1 \boxtimes D_2$ of digraphs $D_1$ and $D_2$, the directed distance $\overrightarrow{d}_{D_1 \boxtimes D_2}((u_i,v_r),(u_j,v_s))=\max\{\overrightarrow{d}_{D_1}(u_i,u_j),\overrightarrow{d}{D_2}(v_r,v_s)\}$.
Thus if $d$ is the metric \textit{maximum distance}, we get
\begin{align*}
d_{D_1 \boxtimes D_2}((u_i,v_r),(u_j,v_s))&=\max\{\overrightarrow{d}_{D_1 \boxtimes D_2}((u_i,v_r),(u_j,v_s)),\overrightarrow{d}_{D_1 \boxtimes D_2}((u_j,v_s),(u_i,v_r))\}\\
&=\max\{\max\{\overrightarrow{d}_{D_1}(u_i,u_j),\overrightarrow{d}_{D_2}(v_r,v_s)\},\max\{\overrightarrow{d}_{D_1}(u_j,u_i),\overrightarrow{d}_{D_2}(v_s,v_r)\}\}\\
&=\max\{\max\{\overrightarrow{d}_{D_1}(u_i,u_j),\overrightarrow{d}_{D_2}(v_r,v_s),\overrightarrow{d}_{D_1}(u_j,u_i),\overrightarrow{d}_{D_2}(v_s,v_r)\}\\     
&=\max\{\max\{\overrightarrow{d}_{D_1}(u_i,u_j),\overrightarrow{d}_{D_1}(u_j,u_i)\},\max\{\overrightarrow{d}_{D_2}(v_r,v_s),\overrightarrow{d}_{D_2}(v_s,v_r)\}\}\\
&=\max\{d_{D_1}(u_i,u_j),d_{D_2}(v_r,v_s)\}\\
\text{Hence it follows that}\\
ecc_{D_1 \boxtimes D_2}(u_i,v_r)&=\max\{d_{D_1 \boxtimes D_2}((u_i,v_r),(u_j,v_s)):(u_j,v_s) \in V(D_1 \boxtimes D_2)\}\\
&=\max\{\max\{d_{D_1}(u_i,u_j),d_{D_2}(v_r,v_s)\}:u_j \in V(D_1),v_s \in V(D_2)\}\\
&=\max\{\max\{d_{D_1}(u_i,u_j):u_j \in V(D_1)\},\max\{d_{D_2}(v_r,v_s):v_s \in V(D_2)\}\\
&=\max\{ecc_{D_1}(u_i),ecc_{D_2}(v_r)\}.\\
\text{Therefore}\\
rad(D_1 \boxtimes D_2)&=\min_{(u_i,v_r) \in V(D_1 \boxtimes D_2)}\{ecc(u_i,v_r)\}\\&=\min_{\substack{u_i \in V(D_1),\\ v_r \in V(D_2)}}\{\max\{ecc_{D_1}(u_i),ecc_{D_2}(v_r)\}\}\\&=\max\{\min_{u_i \in V(D_1)}\{ecc(u_i)\},\min_{v_r \in V(D_2)}\{ecc(v_r)\}\}\\&=\max\{rad(D_1),rad(D_2)\}.\\
\text{Similarly}\\
diam(D_1 \boxtimes D_2)&=\max_{(u_i,v_r) \in V(D_1 \boxtimes D_2)}\{ecc(u_i,v_r)\}\\&=\max_{\substack{u_i \in V(D_1),\\ v_r \in V(D_2)}}\{\max\{ecc_{D_1}(u_i),ecc_{D_2}(v_r)\}\}\\&=\max\{\max_{u_i \in V(D_1)}\{ecc(u_i)\},\max_{v_r \in V(D_2)}\{ecc(v_r)\}\}\\&=\max\{diam(D_1),diam(D_2)\}.
\end{align*}
The strong product of two directed graphs $D_1$ and $D_2$ is strongly connected if and only if both $D_1$ and $D_2$ are strongly connected \cite{handbook}.
Also we can see that $N_{D_1 \boxtimes D_2}[(u_i,v_r)]=N_{D_1}[u_i] \times N_{D_2}[v_r]$.
In \cite{caceres2010boundary}, C\'aceres et al. presented  a description of the boundary type sets of undirected graphs.\\
The description is as follows.  In the case of undirected graphs $G$ and $H$ with diameters $diam(G)$ and $diam(H)$ and radii $rad(G)$ and $rad(H)$ respectively,
$\partial(G \boxtimes H)=(\partial(G) \times V(H)) \bigcup (V(G) \times \partial(H))$.  But we can see that this is not true in the case of digraphs.  See example \ref{exam1}.\\  Here $\partial(D_1)=\{u_1,u_3\}$ and $\partial(D_2)=\{v_1,v_4,v_5\}$.  Unlike undirected graphs, $\partial(D_1 \boxtimes D_2)=\{(u_1,v_1), (u_1,v_3), (u_1,v_4), (u_1,v_5), (u_3,v_1), (u_3,v_3), (u_3,v_4), (u_3,v_5),$\\ $(u_2,v_1), (u_2,v_3), (u_2,v_4), (u_2,v_5)\}$.  That is $(u_1,v_2), (u_3,v_2) \notin \partial(D_1 \boxtimes D_2)$.

\begin{example}\label{exam1}
	\begin{center}
		\begin{figure}[H]
			\psscalebox{.5}
		{	{
				\begin{pspicture}(-4,10)
				\pscircle(-1,9){.035}
				\put(-1.4,8.7){$u_1$}
				\put(-1,8.8){2}
				\pscircle(-1,6){.035}
				\put(-1.4,5.7){$u_2$}
				\put(-1,5.8){1}
				\psline{->}(-1,9)(-1,6)
				\psline{->}(-1,6)(-1,9)
				\pscircle(-1,3){.035}
				\put(-1.4,2.7){$u_3$}
				\put(-1,2.8){2}
				\put(-1.25,1.5){$D_1$}
				\psline{->}(-1,6)(-1,3)
				\psline{->}(-1,3)(-1,6)
				\pscircle(0,1){.035}
				\put(0,.5){$v_1$}
				\put(0,1.1){4}
				\pscircle(2.5,1){.035}
				\put(2.5,.5){$v_2$}
				\put(2.5,1.1){3}
				\pscircle(5,1){.035}
				\put(5,.5){$v_3$}
				\put(5,1.1){2}
				\pscircle(7.5,1){.035}
				\put(7.5,.5){$v_4$}
				\put(7.5,1.1){3}
				\pscircle(10,1){.035}
				\put(10,.5){$v_5$}
				\put(10,1.1){4}
				\put(5,-.25){$D_2$}
				
				\put(5,10.5){$D_1 \boxtimes D_2$}
				\psline{->}(9.95,1)(7.53,1)
				\psline{->}(7.45,1)(5.03,1)
				\psline{->}(5.03,1)(7.45,1)
				\psline{->}(4.95,1)(2.53,1)
				\psline{->}(2.53,1)(4.95,1)
				\psline{->}(2.45,1)(0.05,1)
				\pscircle(0,9){.035}
				\put(0,8.5){$(u_1,v_1)$}
				\put(0,9.1){4}
				\pscircle(0,6){.035}
				\put(0,5.5){$(u_2,v_1)$}
				\put(0,6.1)4
				\psline{->}(0,6)(0,9)
				\psline{->}(0,9)(0,6)
				\pscircle(0,3){.035}
				\put(0,2.5){$(u_3,v_1)$}
				\put(0,3.1){4}
				\psline{->}(0,6)(0,3)
				\psline{->}(0,3)(0,6)
				\pscircle(2.5,9){.035}
				\put(2.5,8.5){$(u_1,v_2)$}
				\put(2.5,9.1){3}
				\psline{->}(2.5,5.97)(2.5,3.1)
				\psline{->}(2.5,3.1)(2.5,5.97)
				\pscircle(2.5,6){.035}
				\put(2.5,5.5){$(u_2,v_2)$}
				\put(2.5,6.1){3}
				\psline{->}(2.5,6.05)(2.5,8.95)
				\psline{->}(2.5,8.97)(2.5,6.1)
				\psline{->}(2.5,6.1)(0,8.97)
				\psline{->}(2.5,8.9)(0.05,6.05)
				\pscircle(2.5,3){.035}
				\put(2.5,2.5){$(u_3,v_2)$}
				\put(2.5,3.1){3}
				\psline{->}(2.5,3.05)(0,5.95)
				\psline{->}(2.5,5.9)(0.05,3.05)
				\pscircle(5,9){.035}
				\put(5,8.5){$(u_1,v_3)$}
				\put(5,9.1){2}
				\psline(5,5.97)(5,3.1)
				\psline(5,3.1)(5,5.97)
				\pscircle(5,6){.035}
				\put(5,5.5){$(u_2,v_3)$}
				\put(5,6.1){2}
				\pscircle(5,3){.035}
				\put(5,2.5){$(u_3,v_3)$}
				\put(5,3.1){2}
				\psline{->}(5,3.05)(2.5,5.95)
				\psline{->}(2.5,5.95)(5,3.05)
				\psline{->}(2.5,8.95)(5,6.05)
				\psline{->}(5,5.95)(7.5,3.05)
				\psline{->}(5,8.95)(7.5,6.05)

				\psline{->}(5,6)(2.5,3)
				\psline{->}(2.5,3)(5,6)
				\psline{->}(2.55,6.05)(5,8.9)
				\psline{->}(5,8.9)(2.55,6.05)
				\psline{->}(2.5,8.95)(5,6.05)
				\psline{->}(5,5.95)(7.5,3.05)
				\psline{->}(5,8.95)(7.5,6.05)
				
				\psline{->}(7.5,6)(5,3)
				\psline{->}(5,3)(7.5,6)
				\psline{->}(5.05,6.05)(7.5,8.9)
				\psline{->}(7.5,8.9)(5.05,6.05)
				\psline{->}(5,6.1)(5,8.97)
				\psline{->}(5,8.97)(5,6.1)
				\psline{->}(5,6.05)(2.5,8.95)
				\pscircle(7.5,9){.035}
				\put(7.5,8.5){$(u_1,v_4)$}
				\put(7.5,9.1){4}
				\pscircle(7.5,6){.035}
				\put(7.5,5.5){$(u_2,v_4)$}
				\put(7.5,6.1){4}
				\pscircle(7.5,3){.035}
				\psline{->}(7.5,6.05)(5,8.95)
				\put(7.5,2.5){$(u_3,v_4)$}
				\put(7.5,3.1){4}
				\psline{->}(7.5,3.05)(5,5.95)
				\psline{->}(7.5,5.97)(7.5,3.1)
				\psline{->}(7.5,3.1)(7.5,5.97)
				\psline{->}(7.5,6)(7.5,9)
				\psline{->}(7.5,9)(7.5,6)
				\pscircle(10,9){.035}
				\put(10,8.5){$(u_1,v_5)$}
				\put(10,9.1){4}
				\pscircle(10,6){.035}
				\put(10,5.5){$(u_2,v_5)$}
				\put(10,6.1){4}
				\pscurve{->}(2.5,8.93)(6.125,8.5)(10,5.93)
				\pscurve{->}(2.5,5.93)(6.125,5.5)(10,2.93)
				
				\pscurve{->}(0,8.9)(2.5,5.5)(5,6.08)
				\pscurve{->}(0,5.9)(2.5,2.5)(5,3.08)

				\pscurve{->}(0,5.9)(2.5,8)(5,8.88)
				\pscurve{->}(0,2.9)(2.5,5)(5,5.95)
				
				\pscurve{->}(2.65,6.1)(5,8)(9.86,8.92)
				\pscurve{->}(2.65,3.1)(5,5)(9.86,5.92)
				
				\pscircle(10,3){.035}
				\put(10,2.5){$(u_3,v_5)$}
				\put(10,3.1){4}
				\psline{->}(10,6)(10,3)
				\psline{->}(10,3)(10,6)
				\pscurve(10,3)(7.5,6)
				\psline{->}(10,6)(10,9)
				\psline{->}(10,9)(10,6)
				\psline{->}(9.9,3)(7.6,3)
				\psline{->}(9.9,6)(7.6,6)
				\psline{->}(9.9,9)(7.6,9)
				
				\psline{->}(9.95,6.1)(7.6,9)
				\psline{->}(9.95,9)(7.6,6.1)
				\psline{->}(9.95,3.1)(7.6,6)
				\psline{->}(9.95,6)(7.6,3.1)
				\psline{->}(7.48,3)(5.1,3)
				
				\psline{->}(7.4,6)(5.1,6)
				\psline{->}(7.4,9)(5.1,9)
				
				\psline{->}(4.98,3)(2.6,3)
				\psline{->}(4.98,6)(2.6,6)
				\psline{->}(4.98,9)(2.6,9)
				
				\psline{->}(2.48,3)(0.1,3)
				\psline{->}(2.48,6)(0.1,6)
				\psline{->}(2.48,9)(0.1,9)
				
				\pscurve{->}(0.05,1)(2.5,0)(5.03,.99)
				\pscurve{->}(2.51,1)(6.125,2)(9.9,1)
				
				\pscurve{->}(0.1,3)(2.5,2)(4.9,3)
				\pscurve{->}(2.6,3)(6.125,5)(9.9,3)
				
				\pscurve{->}(0.1,6)(2.5,5)(4.9,6)
				\pscurve{->}(2.6,6)(6.125,5)(9.9,6)
				
				\pscurve{->}(0.1,9)(2.5,8)(4.95,9)
				\pscurve{->}(2.6,9)(6.125,10)(9.95,9)
				\end{pspicture}
			}}
		\end{figure}
	\end{center}
\end{example}
\begin{theorem}
	Let $D_1$ and $D_2$ be two strongly connected digraphs.  Then
	$\partial(D_1 \boxtimes D_2)=A_1\bigcup A_2\bigcup A_3$, where $A_1=\partial(D_1) \times \partial(D_2)$,
	$A_2 =\{(u_i,v_r)|u_i \in \partial(D_1), v_r \notin \partial(D_2), d(v_r,v_s) \leq ecc(u_i), \text{for all } v_s \in N(v_r)\}$ and\\
	$A_3=\{(u_i,v_r)|u_i \notin \partial(D_1), v_r \in \partial(D_2), d(u_i,u_p) \leq ecc(v_r), \text{for all } u_p \in N(u_i)\}$.
\end{theorem}
	\proof
	Suppose that $(u_i,v_r) \in \partial(D_1 \boxtimes D_2)$. Then there exists a vertex $(u_j,v_s) \in V(D_1 \boxtimes D_2)$ such that for all vertices $(u_k,v_q) \in N[(u_i,v_r)]$, $d((u_j,v_s),(u_i,v_r)) \geq d((u_j,v_s),(u_k,v_q))$.  Since $d((u_j,v_s),(u_i,v_r))=\max\{d(u_j,u_i),d(v_s,v_r)\}$, there are three possibilities.
	\begin{enumerate}
		\item
		$d(u_j,u_i) \geq d(u_j,u_k)$, for all $u_j \in N[u_i]$ and $d(v_s,v_r) \geq d(v_s,v_q))$, for all $v_q \in N[v_r]$.
		
		\item
		Only $d(u_j,u_i) \geq d(u_j,u_k)$, for all $u_j \in N[u_i]$ holds. 
		\item
		Only $d(v_s,v_r) \geq d(v_s,v_q)$, for all $v_q \in N[v_r]$ holds.
	\end{enumerate}
	In the first case, $u_i \in \partial(D_1)$ and $v_r \in \partial(D_2)$ and hence $(u_i,v_r) \in A_1=\partial(D_1) \times \partial(D_2)$.\\  
	In the second case, as $v_r \notin \partial(D_2)$, for all $v_s \in V(D_2)$, there exists $v_q \in N(v_r)$ such that $d(v_s,v_r)<d(v_s,v_q)$.  
	Then $d((u_j,v_s),(u_i,v_r))=\max\{d(u_j,u_i),d(v_s,v_r)\}=d(u_j,u_i)$, and $d(u_j,u_i) > d(v_s,v_r)$ for if $d((u_j,v_s),(u_i,v_r))=d(v_s,v_r)$, we get a contradiction since then $d((u_j,v_s),(u_k,v_q))=d(v_s,v_q) > d(v_s,v_r)$ so that $(u_i,v_r)$ cannot be a boundary vertex of $(u_j,v_s)$.  
	Thus we have $d((u_j,v_s),(u_i,v_r))=d((u_j,v_r),(u_i,v_r))=d(u_j,u_i)$.\\ 
	Now if $d(v_r,v_q)>d(u_j,u_i)$, for some $v_q \in N(v_r)$, then $d((u_j,v_r),(u_i,v_q))=\max\{d(u_j,u_i),d(v_r,v_q)\}=d(v_r,v_q)$, which is a contradiction since then $(u_i,v_r)$ could not be the boundary vertex of $(u_j,v_r)$ and hence that of any $(u_j,v_s)$.  Thus necessarily $d(v_r,v_q) \leq d(u_j,u_i)$ for all $v_q \in N(v_r)$.  Thus in this case, for $(u_i,v_r)$ to be a boundary vertex in $D_1 \boxtimes D_2$, it is necessary that $d(v_r,v_q) \leq ecc(u_i)$ for all $v_q \in N(v_r)$ since then $(u_i,v_r)$ will be a boundary vertex of $(u_b,v_r)$ where $u_b \in V(D_1)$ is such that $d(u_b,u_i)=ecc(u_i)$.  So in the second case, $(u_i,v_r) \in A_2$.\\
	Similarly in the third case, we can prove that if $v_r \in \partial(D_2)$ and $u_i \notin \partial(D_1)$, then $(u_i,v_r) \in \partial(D_1 \boxtimes D_2)$ implies that $d(u_i,u_p) \leq ecc(v_r), \text{for all } u_p \in N(u_i)$.  So in the third case, $(u_i,v_r) \in A_3$.
	Hence $\partial(D_1 \boxtimes D_2) \subseteq A_1 \bigcup A_2 \bigcup A_3$.
	
	Conversely, suppose that $(u_i,v_r) \in A_1 \bigcup A_2 \bigcup A_3$.
	First let $(u_i,v_r) \in A_1$.  Then $u_i \in \partial(D_1)$ and $v_r \in \partial(D_2)$.  So there exists a vertex $u_j \in V(D_1)$ such that $d(u_j,u_i) \geq d(u_j,u_k)$ for every $u_k \in N[u_i]$ and there exists a vertex $v_s \in V(D_2)$ such that $d(v_s,v_r) \geq d(v_s,v_q)$ for every $v_q \in N[v_r]$.  Hence in $D_1 \boxtimes D_2$, $d((u_j,v_s),(u_i,v_r))=\max\{d(u_j,u_i),d(v_s,v_r)\} \geq \max\{d(u_j,u_k),d(v_s,v_q)\}=d((u_j,v_s),(u_k,v_q))$ for all vertices $(u_k,v_q) \in N[(u_i,v_r)]$.  Thus $A_1 \subseteq \partial(D_1 \boxtimes D_2)$.
	
	Now let $(u_i,v_r) \in A_2$.  Then  $u_i \in \partial(D_1), v_r \notin \partial(D_2)$ and $d(v_r,v_q) \leq ecc(u_i), \text{for all } v_q \in N[v_r]$.  Since $u_i \in \partial(D_1)$, there exists atleast one vertex $u_j \in V(D_1)$ such that $d(u_j,u_i) \geq d(u_j,u_k)$ for every $u_k \in N[u_i]$.  Of these vertices, let $u_b$ a vertex such that $d(u_b,u_i)=ecc(u_i)$.  Hence in $D_1 \boxtimes D_2$, $d((u_b,v_r),(u_i,v_r))=d(u_b,u_i) =ecc(u_i) \geq d(v_r,v_q)$ for all $v_q \in N[v_r]$.
	Also $d(u_b,u_i) \geq  d(u_b,u_k)$ for all $u_k \in N[u_i]$.  Hence $d((u_b,v_r),(u_i,v_r)) \geq d((u_b,v_r),(u_k,v_q))$ for all vertices $(u_k,v_q) \in N[(u_i,v_r)]$.  Thus $(u_i,v_r)$ is a boundary vertex of $(u_b,v_r)$ in $D_1 \boxtimes D_2$ and hence $A_2 \subseteq \partial(D_1 \boxtimes D_2)$.\\
		Similarly let $(u_i,v_r) \in A_3$.  Then  $u_i \notin \partial(D_1), v_r \in \partial(D_2)$ and $d(u_i,u_p) \leq ecc(v_r), \text{for all } u_p \in N[u_i]\}$.  Since $v_r \in \partial(D_2)$, there exists atleast one vertex $v_s \in V(D_1)$ such that $d(v_s,v_r) \geq d(v_s,v_q)$ for all $v_q \in N[v_r]$.  Of these vertices, let $v_c$ be a vertex such that $d(v_c,v_r)=ecc(v_r)$.  Hence in $D_1 \boxtimes D_2$, $d((u_i,v_c),(u_i,v_r))=d(v_c,v_r) =ecc(v_r) \geq d(u_i,u_k)$ for all $u_k \in N[u_i]$.
	Also $d(v_c,v_r) \geq  d(v_c,v_q)$ for all $v_q \in N[v_r]$.  Hence $d((u_i,v_c),(u_i,v_r)) \geq d((u_i,v_c),(u_k,v_q))$ for all vertices $(u_k,v_q) \in N[(u_i,v_r)]$.  Thus $(u_i,v_r)$ is a boundary vertex of $(u_i,v_c)$ in $D_1 \boxtimes D_2$ and hence $A_3 \subseteq \partial(D_1 \boxtimes D_2)$.\\
	Thus we get 
	$A_1 \bigcup A_2 \bigcup A_3 \subseteq \partial(D_1 \boxtimes D_2)$.\qed
The results concerning the periphery, eccentricity and contour of the strong product of two digraphs are the same as that of undirected graphs which are described in \cite{caceres2010boundary}.  Here we provide their proofs to cover the directed case.
\begin{proposition}
	a) If $diam(D_1) < diam(D_2)$, then $Per(D_1 \boxtimes D_2) =V(D_1) \times Per(D_2)$.\\
	b) If $diam(D_1) = diam(D_2)$, then $Per(D_1 \boxtimes D_2) =Per(D_1) \times V(D_2) \bigcup V(D_1) \times Per(D_2)$.
\end{proposition}
\begin{proof}
	a) Let $diam(D_2)=n$.  Let $v_r \in Per(D_2)$.  Then for all $u_i \in V(D_1)$, $ecc(u_i,v_r)=\max\{ecc(u_i), ecc(v_r)\}=n$.  Hence $(u_i,v_r) \in Per(D_1 \boxtimes D_2)$.  Also if $v_r \notin Per(D_2)$, then since $ecc(u_i,v_r) <n$, $(u_i,v_r) \notin Per(D_1 \boxtimes D_2)$.  Thus $Per(D_1 \boxtimes D_2) =V(D_1) \times Per(D_2)$.\\
	b) Let $diam(D_1)=diam(D_2)=n$.  Let $u_i \in Per(D_1)$.  Then for all $v_r \in V(D_2)$, $(u_i,v_r) \in Per(D_1 \boxtimes D_2)$, as $ecc(u_i,v_r)=\max\{ecc(u_i), ecc(v_r)\}=n$.  Hence $(u_i,v_r) \in Per(D_1 \boxtimes D_2)$.
	Also if $v_r \in Per(D_2)$, then $(u_i,v_r) \in Per(D_1 \boxtimes D_2)$ for all $u_i \in V(D_1)$.  Thus $Per(D_1) \times V(D_2) \bigcup V(D_1) \times Per(D_2) \subseteq Per(D_1 \boxtimes D_2)$.
	Now if $(u_i,v_r) \in Per(D_1 \boxtimes D_2)$, then $ecc(u_i,v_r)=\max\{diam(D_1), diam(D_2)\}=n$.  Thus necessarily atleast one of $ecc(u_i)$ or $ecc(v_r)$ must be equal to $n$.  Hence either $u_i \in Per(D_1)$ or $v_r \in Per(D_2)$. 
	So we get $Per(D_1 \boxtimes D_2) \subseteq Per(D_1) \times V(D_2) \bigcup V(D_1) \times Per(D_2)$.
\end{proof}
\begin{proposition}
	Let $D_1$ and $D_2$ be two strongly connected digraphs.  Then
	\begin{enumerate}
		\item
		If $rad(D_1)=rad(D_2)$, then $Ecc(D_1 \boxtimes D_2)=[Ecc(D_1) \times V(D_2)] \bigcup [V(D_1) \times Ecc(D_2)]$.
		\item
		If $rad(D_1)< rad(D_2)$, then $Ecc(D_1 \boxtimes D_2)=[\bigcup_{u_i \geq r_ {D_2}}Ecc(u_i) \times V(D_2)] \bigcup [V(D_1) \times Ecc(D_2)]$.
	\end{enumerate}
\end{proposition}
\begin{proof}
	\begin{enumerate}
		\item
		First we will prove that $Ecc(D_1 \boxtimes D_2) \subseteq[Ecc(D_1) \times V(D_2)] \bigcup [V(D_1) \times Ecc(D_2)]$.\\Let $(u_i,v_r) \in Ecc(D_1 \boxtimes D_2)$.  Then there exists a vertex $(u_j,v_s)$ such that $ecc(u_j,v_s)=d((u_j,v_s),(u_i,v_r))=\max\{d(u_j,u_i),d(v_s,v_r)\}$.  Since $ecc(u_j,v_s)=\max\{ecc(u_j),ecc(v_s)\}$, and $ecc(u_j) \geq d(u_j,u_i)$ and $ecc(v_s) \geq d(v_s,v_r)$, atleast one of $ecc(u_j)=d(u_j,u_i)$ and $ecc(v_s)=d(v_s,v_r)$ must hold.  So either $u_i$ is an eccentric vertex of $u_j$ or
		$v_r$ is an eccentric vertex of $v_s$.  Hence $(u_i,v_r) \in [Ecc(D_1) \times V(D_2)] \bigcup [V(D_1) \times Ecc(D_2)]$.
		
		Let $rad(D_1)=rad(D_2)=n$.  Let $u_i \in Ecc(D_1)$.  So there exists a vertex $u_j \in V(D_1)$ such that $ecc(u_j)=d(u_j,u_i)$.  Consider $(u_i,v_r) \in V(D_1 \boxtimes D_2)$, where $v_r$ is an arbitrary vertex in $D_2$.  Since $rad(D_2)=n$, there exists a vertex  $v_s \in V(D_1)$ such that $ecc(v_s)=n$.  Hence $d(v_s,v_r) \leq n$ and so $ecc(u_j,v_s)=\max\{ecc(u_j),ecc(v_s)\}= \max\{ecc(u_j),n\}=ecc(u_j)$.  Thus we have $d((u_j,v_s),(u_i,v_r))=\max\{d(u_j,u_i),(v_s,v_r)\}=ecc(u_j)=ecc(u_j,v_s)$.  So $(u_i,v_r)$ is an eccentric vertex of $(u_j,v_s)$.  Thus if  $u_i \in Ecc(D_1)$, then $(u_i,v_r) \in Ecc(D_1 \boxtimes D_2)$ for all $v_r \in V(D_2)$.  Similarly, we can prove that if $v_q \in Ecc(D_2)$, then $(u_k,v_q) \in Ecc(D_1 \boxtimes D_2)$ for all $u_k \in V(D_1)$.  Thus $[Ecc(D_1) \times V(D_2)] \bigcup [V(D_1) \times Ecc(D_2)] \subseteq Ecc(D_1 \boxtimes D_2)$ and so the result holds.
		\item
		$rad(D_1)<rad(D_2)=n$.  Let $u_i \in V(D_1)$.  Here there arise two cases.  Either $v_r \in Ecc(D_2)$ or $v_r \notin Ecc(D_2)$.\\First suppose that $v_r \in Ecc(D_2)$.  Then there exists a vertex $v_s \in V(D_2)$ such that $ecc(v_s)=d(v_s,v_r)$.  We have a vertex $u_p \in V(D_1)$ such that $ecc(u_p)=rad(D_1)$.  Then since $rad(D_2) >ecc(u_p)$, we get $ecc(u_p,v_s)=\max\{ecc(u_p),ecc(v_s)\}=ecc(v_s)$.  Also, $d((u_p,v_s),(u_i,v_r))=\max\{d(u_p,u_i),d(v_s,v_r)\}=ecc(v_s)$.  Thus $(u_i,v_r)$ is an eccentric vertex of $(u_p,v_s)$.  So in this case, we have $V(D_1) \times Ecc(D_2) \subseteq Ecc(D_1 \boxtimes D_2)$.\\ 
		Now suppose that $v_r \notin Ecc(D_2)$.  Let $v_q \in V(D_2)$ be such that $ecc(v_q)=rad(D_2)$.  
		Take $\bigcup_{u_i \geq r_ {D_2}}Ecc(u_i)=A$.  Let $u_k \in A$.  Then there exists a vertex $u_p \in V(D_1)$ such that $ecc(u_p) \geq rad(D_2)$ and $ecc(u_p)=d(u_p,u_k)$.  Then $d((u_p,v_q),(u_k,v_r))=\max\{d(u_p,u_k),d(v_q,v_r)\}=d(u_p,u_k)=ecc(u_p)=ecc(u_p,v_q)$ and hence $(u_k,v_r)$ is an eccentric vertex of $(u_p,v_q)$. 
		So here we get $\bigcup_{u_i \geq r_ {D_2}}Ecc(u_i) \times V(D_2) \subseteq Ecc(D_1 \boxtimes D_2)$.\\
		Thus $[\bigcup_{u_i \geq r_ {D_2}}Ecc(u_i) \times V(D_2)] \bigcup [V(D_1) \times Ecc(D_2)] \subseteq Ecc(D_1 \boxtimes D_2)$.\\
		Conversely, let $(u_k,v_r) \in Ecc(D_1 \boxtimes D_2)$.  Then there exists a vertex $(u_j,v_s) \in V(D_1 \boxtimes D_2)$ such that $ecc(u_j,v_s)=d((u_j,v_s),(u_k,v_r))=\max\{d(u_j,u_k),(v_s,v_r)\}=\max\{ecc(u_j),ecc(v_s)\}$.  If $v_r \in Ecc(D_2)$, we get $(u_k,v_r) \in V(D_1) \times Ecc(D_2)$.  So let $(u_k,v_r) \in Ecc(D_1 \boxtimes D_2)$ and $v_r \notin Ecc(D_2)$.  Then for all $v_s \in V(D_2)$, $ecc(v_s) > d(v_s,v_r)$.  Hence $ecc(u_j,v_s)=ecc(u_j)=d(u_j,u_k)$.  So if possible, suppose that $u_k \notin A$.  Thus there is no vertex $u_i$ such that $ecc(u_i)=d(u_i,u_k)$ and $ecc(u_i) \geq rad(D_2)$.  Hence for all vertices $u_i$ such that $ecc(u_i)=d(u_i,u_k)$, $d(u_k,u_i) < rad(D_2)$.  Then for all $(u_i,v_s) \in V(D_1 \boxtimes D_2)$, $d((u_i,v_s),(u_k,v_r))=\max\{d(u_i,u_k),(v_s,v_r)\}=d(v_s,v_r)<ecc(v_s) \leq ecc(u_i,v_s)$.  This contradicts our assumption that $(u_k,v_r) \in Ecc(D_1 \boxtimes D_2)$.  Hence $u_k \in A$.
		So in this case, we get $(u_i,v_r) \in \bigcup_{u_i \geq r_ {D_2}}Ecc(u_i) \times V(D_2)$.  Hence $Ecc(D_1 \boxtimes D_2) \subseteq [\bigcup_{u_i \geq r_ {D_2}}Ecc(u_i) \times V(D_2)] \bigcup [V(D_1) \times Ecc(D_2)]$ and so the result holds.
	\end{enumerate}
\end{proof}
\begin{proposition}
	Let $D_1$ and $D_2$ be two strongly connected digraphs.  Then $Ct(D_1 \boxtimes D_2)=\{(u_i,v_r) \in V(D_1 \boxtimes D_2): u_i \in Ct(D_1), ecc(v_r)<ecc(u_i)\} \bigcup \{(u_i,v_r) \in V(D_1 \boxtimes D_2): v_r \in Ct(D_2), ecc(u_i)<ecc(v_r)\} \bigcup [Ct(D_1) \times Ct(D_2)]$.
\end{proposition}
\begin{proof}
	
	$(u_i,v_r) \in Ct(D_1 \boxtimes D_2)$ if and only if $ecc(u_i,v_r) \geq ecc(u_k,v_q)$ for all $(u_k,v_q) \in N[(u_i,v_r)]$;  if and only if $\max\{ecc(u_i),ecc(v_r)\} \geq \max\{ecc(u_k),ecc(v_q)\}$ for all $u_k \in N[u_i]$ and  $v_q \in N[v_r]$;  if and only if one of the following three cases holds.
		\begin{enumerate}
		\item 
		$ecc(u_i) \geq ecc(u_k)$ and $ecc(v_r) \geq ecc(v_q)$ for all $u_k \in N[u_i]$ and  $v_q \in N[v_r]$.
		\item 
		$ecc(v_r) < ecc(u_i)$ and $ecc(u_i) \geq ecc(u_k)$ for all $u_k \in N[u_i]$.
		\item 
		$ecc(u_i) < ecc(v_r)$ and $ecc(v_r) \geq ecc(v_q)$ for all $v_q \in N[v_r]$.
	\end{enumerate}
	Thus we get $Ct(D_1 \boxtimes D_2)=\{(u_i,v_r) \in V(D_1 \boxtimes D_2): u_i \in Ct(D_1), ecc(v_r)<ecc(u_i)\} \bigcup \{(u_i,v_r) \in V(D_1 \boxtimes D_2): v_r \in Ct(D_2), ecc(u_i)<ecc(v_r)\} \bigcup [Ct(D_1) \times Ct(D_2)]$.
\end{proof}

\section{Conclusion}

In  the study of large networks which can be represented by strongly connected  digraphs, the determination of boundary-type sets has important applications.  The boundary-type sets of  almost all large directed networks can be determined by combining the unique prime factor decomposition and the results obtained relating to boundary-type sets of strong product of digraphs and that of its factors.    This information can be used to determine the efficiency of the network in physical, biological and social set ups.
\section*{Acknowledgements}
Prasanth G. Narasimha-Shenoi and Mary Shalet T. J are supported by Science and Engineering Research Board, a statutory board of Government of India under their Extra Mural Research Funding No. EMR/2015/002183. Also, their research was partially supported by Kerala State Council for Science Technology and Environment of Government of Kerala under their SARD project grant Council(P) No. 436/2014/KSCSTE.  Prasanth G. Narasimha-Shenoi is also supported by  Science and Engineering Research Board, a statutory board of Government of India under their MATRICS Scheme No. MTR/2018/000012.
\bibliographystyle{splncs04}
\bibliography{caldam2020}

\begin{thebibliography}{10}
\providecommand{\url}[1]{\texttt{#1}}
\providecommand{\urlprefix}{URL }
\providecommand{\doi}[1]{https://doi.org/#1}

\bibitem{bang2008digraphs}
Bang-Jensen, J., Gutin, G.Z.: Digraphs: theory, algorithms and applications.
  Springer Science \& Business Media (2008)

\bibitem{caceres2006geodetic}
C{\'a}ceres, J., Hernando, C., Mora, M., Pelayo, I.M., Puertas, M.L., Seara,
  C.: On geodetic sets formed by boundary vertices. Discrete Mathematics
  \textbf{306}(2),  188--198 (2006)

\bibitem{caceres2010boundary}
C{\'a}ceres, J., Hernando~Mart{\'\i}n, M.d.C., Mora~Gin{\'e}, M.,
  Pelayo~Melero, I.M., Puertas~Gonz{\'a}lez, M.L.: Boundary-type sets and
  product operators in graphs. In: VII Jornadas de Matem{\'a}tica Discreta y
  Algor{\'\i}tmica. pp. 187--194 (2010)

\bibitem{caceres2005rebuilding}
C{\'a}ceres~Gonz{\'a}lez, J., M{\'a}rquez~P{\'e}rez, A., Oellermann, O.R.,
  Puertas~Gonz{\'a}lez, M.L.: Rebuilding convex sets in graphs. Discrete
  Mathematics, 297 (1-3), 26-37.  (2005)

\bibitem{cartesian}
Changat, M., Narasimha-Shenoi, P.G., Joseph, M.S.T., Kumar, R.: Boundary
  vertices of cartesian product of directed graphs. International Journal of
  Applied and Computational Mathematics  \textbf{5}(1), ~19 (2019)

\bibitem{chartrand2003boundary}
Chartrand, G., Erwin, D., Johns, G.L., Zhang, P.: Boundary vertices in graphs.
  Discrete Mathematics  \textbf{263}(1-3),  25--34 (2003)

\bibitem{chartrand1997distance}
Chartrand, G., Tian, S.: Distance in digraphs. Computers \& Mathematics with
  Applications  \textbf{34}(11),  15--23 (1997)

\bibitem{crespelle}
Crespelle, C., Thierry, E., Lambert, T.: A linear-time algorithm for computing
  the prime decomposition of a directed graph with regard to the cartesian
  product. In: International computing and combinatorics conference. pp.
  469--480. Springer (2013)

\bibitem{feigenbaum1986directed}
Feigenbaum, J.: Directed cartesian-product graphs have unique factorizations
  that can be computed in polynomial time. Discrete applied mathematics
  \textbf{15}(1),  105--110 (1986)

\bibitem{handbook}
Hammack, R., Imrich, W., Klav{\v{z}}ar, S.: Handbook of product graphs. CRC
  press (2011)

\bibitem{hammack18}
Hammack, R.H.: Digraphs products. In: Classes of Directed Graphs, pp. 467--515.
  Springer (2018)

\bibitem{hellmuth}
Hellmuth, M., Marc, T.: On the cartesian skeleton and the factorization of the
  strong product of digraphs. Theoretical Computer Science  \textbf{565},
  16--29 (2015)

\bibitem{net}
Newman, M.: Networks : An Introduction. Oxford University Press, Oxford (2010)

\bibitem{robbins1939theorem}
Robbins, H.E.: A theorem on graphs, with an application to a problem of traffic
  control. The American Mathematical Monthly  \textbf{46}(5),  281--283 (1939)

\end{thebibliography}
\end{document}